\documentclass[12pt,onecolumn,draftclsnofoot]{IEEEtran}

\usepackage{cite}	
\usepackage{graphicx} 
\usepackage[latin1]{inputenc} 
\usepackage[T1]{fontenc} 
\usepackage{amsmath,amsfonts,amsbsy,amssymb} 
\usepackage{mathabx} 
\usepackage{amssymb} 
\usepackage{amsmath}
\usepackage{amsthm}	
\usepackage{mathrsfs}
\usepackage[nolist]{acronym} 
\usepackage{tabularx} 
\usepackage{multirow}
\usepackage{wasysym}
\usepackage{float}
\usepackage{color} 

\usepackage{enumitem} 

\newtheorem{proposition}{Proposition}
\newtheorem{lemma}{Lemma}

\begin{document}

\title{On the Effective Capacity of MTC Networks in the Finite Blocklength Regime} 
\author{Mohammad Shehab, Endrit Dosti, Hirley Alves, and Matti Latva-aho\\
	
	\IEEEauthorblockA{
		Centre for Wireless Communications (CWC), University of Oulu, Finland\\
	}
	\vspace{-2mm}
	Email: firstname.lastname@oulu.fi
}
%
\maketitle


\begin{abstract}
This paper analyzes the effective capacity (EC) of delay constrained machine type communication (MTC) networks operating in the finite blocklength (FB) regime. First, we derive a closed-form mathematical approximation for the EC in Rayleigh block fading channels. We characterize the optimum error probability to maximize the concave EC function and study the effect of SINR variations for different delay constraints. Our analysis reveals that SINR variations have less impact on EC for strict delay constrained networks. We present an exemplary scenario for massive MTC access to analyze the interference effect proposing three methods to restore the EC for a certain node which are power control, graceful degradation of delay constraint and joint compensation. Joint compensation combines both power control and graceful degradation of delay constraint, where we perform maximization of an objective function whose parameters are determined according to delay and SINR priorities. Our results show that networks with stringent delay constraints favor power controlled compensation and compensation is generally performed at higher costs for shorter packets.
\end{abstract}
\vspace{-1mm}


\section{Introduction}\label{introduction}
Nowadays, the road to the fifth generation of mobile communication is being paved by leaps and bounds. 5G should be able to support new features such as ultra reliable transmission and massive machine-to-machine (M2M) communication \cite {paper1}. MTC has gained an increasing interest in recent years \cite{Orsino2017}. MTC networks are expected to connect massive number of devices which communicate with high reliability and minimum latency to support mission critical applications, envisaged as \textit{Ultra-Reliable Communication} (URC) \cite{paper1,Dosti,paper3,NokiacMTC2016}. Furthermore, cooperative transmission in MTC is a hot topic for research so far \cite{MTC2}.

Unlike traditional communication systems, which are based on infinite coding schemes, MTC networks are designed to communicate on short packets and delay limited basis. Such demands stimulated a revolutionary trend in information theory studying communication at finite blocklength (FB) \cite {paper2,paper5,paper1,paper14, paper3}. In \cite {paper3}, a per-node throughput model was introduced for AWGN and quasi-static collision channels where interference is treated as additive Gaussian noise while considering average delay. In \cite{paper14}, Yang et al. characterized the attainable rate as a function of blocklength and error probability $\epsilon$ for block fading channels.
 
To model the delay requirements in MTC networks, we resort to the effective capacity metric which was introduced in \cite{paper6} to provide an indication of the maximum possible arrival rate that can be supported by a network subject to a particular latency requirement. A statistical model for a single node effective rate in bits per channel use (bpcu) for a certain error probability with delay exponent was discussed in \cite{paper5} for Rayleigh block fading channels where the channel coefficients are assumed to be fixed for a block of $T_f$ symbols. However, a closed-form expression for the EC was not provided in these discussions. Latency-throughput tradeoff was characterized in \cite{Nokia2} for cellular networks exploiting the EC theory. Musavian et al. analyzed the EC maximization of secondary node with some interference power constraints for primary node in a cognitive radio environment with interference constraints \cite{paper9}. To the best of our knowledge, EC for FB packets transmission in multi-node MTC scenario has not been investigated. For convenience in this paper, we refer to one machine terminal as node.

In this paper, we derive a mathematical expression for EC in quasi-static Rayleigh fading for delay limited networks. This leads us to characterize the optimum error probability which maximizes the EC. We consider dense MTC networks and characterize the effect of interference on their EC. We propose three methods to allow a certain node maintain its EC which are: \textit{i}) Power control; \textit{ii}) graceful degradation of delay constraint; and \textit{iii}) joint model. Power control depends on increasing the power of a certain node to recover its EC which in turn degrades the SINR of other nodes. Our analysis proves that SINR variations have limited effect on EC in networks with stringent delay limits. Hence, the side effect of power control is worse for less stringent delay constraints and vice versa. We illustrate the trade off between power control and graceful degradation of delay constraint. Furthermore, we introduce a joint model which combines both of them. The operational point to determine the amount of compensation performed by each of the two methods in the joint model is determined by maximization of an objective function leveraging the network performance.
\vspace{-1mm}
\section{system layout} \label{system model}
\vspace{-1mm}
\subsection{Network model}
\vspace{-1mm}
We consider a transmission scenario in which $N$ nodes  transmit packets with equal power to a common controller through a Rayleigh block fading collision channel with blocklength $T_f$ as shown in Fig. \ref{channel}. The received vector $\mathbf {y}_n\in \mathbb{C}^n$ of node $n$ is given by
\begin{align}\label{eq1}
\mathbf {y}_n=h_n\mathbf {x}_n+\sum_{s\neq n} h_s\mathbf {x}_s+\mathbf {w},
\end{align}
where $\mathbf {x}_n \in \mathbb{C}^n$ is the transmitted packet of node $n$, $h_n$ is the fading coefficient for node $n$ which is assumed to be quasi-static with Rayleigh distribution and thus, remains constant over $T_f$ symbols which span the whole packet duration. The index $s$ includes all $N-1$ interfering nodes which collide with node $n$, and $\mathbf{w}$ is the additive complex Gaussian noise vector whose entries are of unit variance. Given the SNR $\rho$ of a single node, the SINR of any node $n$ is
\begin{align}\label{eq2}
\rho_i=\frac{\rho}{1+\rho\sum_{s} |h_s|^2}.
\end{align}
\begin{figure}[!t] 
	\centering
	\includegraphics[width=0.7\columnwidth]{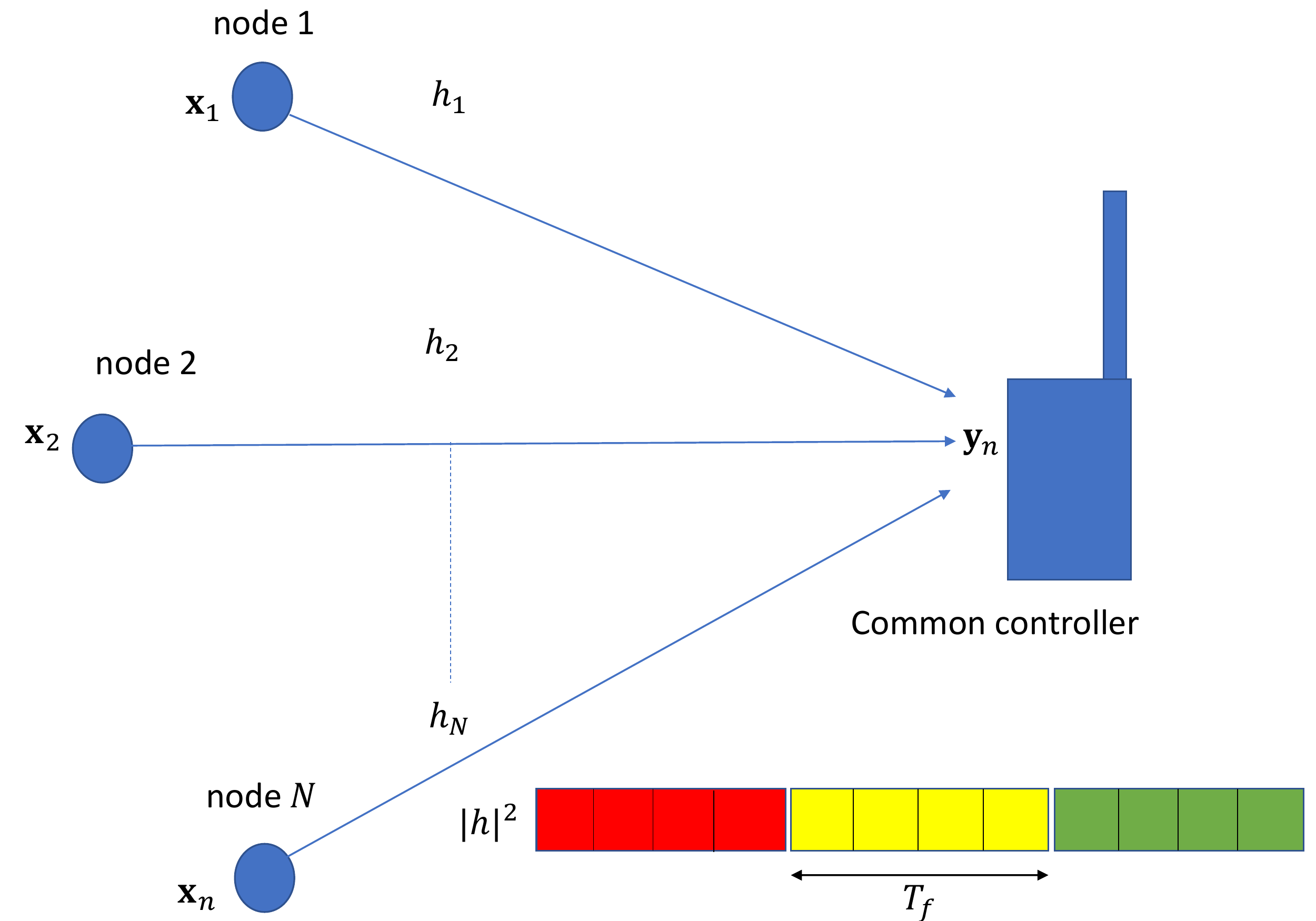}
	\vspace{-2mm}
	\caption{Network Layout.}
	\label{channel}
	\vspace{-6mm}
\end{figure}
We assume that all nodes are equidistant from the common controller and CSI is available at each node and the controller. Thus, as the number of nodes increases, the sum of Rayleigh distributed fading envelopes of $N-1$ interfering nodes becomes $\sum_{s} |h_s|^2\approx N-1$ \cite{fundamentals} and the interference resulting from nodes in set $s$ can be modeled as in \cite{paper3} where (\ref{eq2}) reduces to 
\begin{align}\label{eq2.2}
\rho_i=\frac{\rho}{1+\rho \ (N-1)}.
\end{align}
Note that, CSI acquisition in this setup is not trivial and its cost is negligible whenever the channel remains constant over multiple symbols. Additionally, as in \cite{paper3} we aim to provide a performance benchmark for such networks without interference coordination. 
\vspace{-1mm}
\subsection{Communication at Finite Blocklength}
We start by presenting the notion of FB transmission, in which short packets are conveyed at rate that depends not only on the SNR, but also on the blocklength and the probability of error $\epsilon$ \cite{paper2}. In this case, $\epsilon$ has a small value but not vanishing. For error probability $\epsilon \in\left[ 0,1\right] $, the normalized achievable rate in bpcu is given by 
\begin{align}\label{eq3}
\begin{split}
r\approx&\log_2(1+\rho_i|h|^2) \\
&-\sqrt{\frac{1}{T_f}\left( 1-\frac{1}{\left( 1+\rho_i|h|^2\right) ^{2}}\right) } Q^{-1}
(\epsilon)\log_2(e),	
\end{split}
\end{align}
where $Q(\cdot)=\int_{\cdot}^{\infty}\frac{1}{\sqrt{2 \pi}}e^{\frac{-t^2}{2}} dt$ is the Gaussian Q-function, and $Q^{-1} (\cdot)$ represents its inverse, $\rho_i$ is the SINR and $|h|^2$ is the fading envelope.  

\section{Effective Capacity analysis under Finite Blocklength} \label{EC FB} 
The concept of EC indicates the capability of communication nodes to exchange data with maximum rate and certain latency constraint. An outage occurs when a packet delay exceeds a maximum delay bound $D_{max}$ and its probability is defined as \cite{paper6}
\begin{align}\label{delay}
P_{out\_ delay}=Pr(delay \geq D_{max}) \approx e^{-\theta \cdot EC \cdot D_{max}}.
\end{align}	 
Conventionally, a network's tolerance to long delay is measured by the delay exponent $\theta$. The network has more tolerance to large delays for small values of $\theta$ (i.e., $\theta\rightarrow 0$), while for large values of $\theta$, it becomes more delay strict. For example, a network with unity EC and an outage probability of $10^{-3}$ can tolerate a maximum delay of $691$ symbol periods for $\theta=0.01$ and $23$ symbol periods when $\theta=0.3$. In quasi-static fading, the channel remains constant within each transmission period $T_f$ \cite{paper13}, and the EC is \cite{paper5}
\begin{align}\label{EC}
EC(\rho_i,\theta,\epsilon)=-\frac{1}{T_f\theta} \ln\left(E_{z=|h|^2}\left[\epsilon+(1-\epsilon)e^{-T_f\theta r}\right]\right).
\end{align} 
\begin{lemma} \label{L1}
	The effective capacity of a certain node communicating in a Rayleigh block fading channel is given by
\begin{align}\label{Rayleigh}
\begin{split}
EC(\rho_i,\theta,\epsilon)\approx-\frac{1}{T_f\theta} \ln \left[\epsilon+(1-\epsilon) \  \mathcal{J}\right],
\end{split}
\end{align}
with
\begin{align}\label{J}
\begin{split}
\mathcal{J}=\sum_{m=0}^M c^m \int_{0}^{\infty}(1+\rho_i z)^{d} \frac{x^m}{m!}e^{-z} dz,
\end{split}
\end{align}
where $d=\frac{-\theta T_f}{\ln(2)}$. Also let $c=\theta \sqrt{T_f} Q^{-1}(\epsilon)\log_2e$ and $x=\sqrt{(1-\frac{1}{(1+\rho_i z)^{2}})}$.
\end{lemma}
\begin{proof}
For Rayleigh envelope of pdf $f_{|h|^2}(z)=e^{-z}$, the EC expression in (\ref{EC}) can be written as
\begin{align}\label{EC2}
\begin{split}
EC(\rho_i,\theta,\epsilon)= -\frac{1}{T_f\theta} \ln\left(\int_{0}^{\infty}
\left( \epsilon+(1-\epsilon)e^{-\theta T_f r}\right)  e^{-z} dz\right).  
\end{split}
\end{align}
From (\ref{eq3}), we have		
\begin{align}\label{e1}
e^{-\theta T_f r}=e^{-\theta T_f \log_2(1+\rho_i z)}e^{\theta \sqrt{T_f(1-\frac{1}{(1+\rho_i z)^{2}})} Q^{-1}(\epsilon)\log_2e}.
\end{align}	
Elaborating, we attain
\begin{flalign}\label{e2}
e^{-\theta T_f \log_2(1+\rho_i z)}
&=(1+\rho_i z)^{d},
\end{flalign}
\begin{align}\label{e3}
e^{\theta \sqrt{T_f(1-\frac{1}{(1+\rho_i z)^{2}})} Q^{-1}(\epsilon)\log_2e}=e^{cx},  
\end{align}
where $e^{cx} = \sum_{m=0}^\infty \frac{(cx)^m}{m!}$. It follows from (\ref{e1}), (\ref{e2}) and (\ref{e3}) that the expression in (\ref{EC2}) can be written as a truncated sum of $m$ terms leading to (\ref{Rayleigh})  and the approximation becomes equality as $M\rightarrow \infty$.
\end{proof}
The infinite series in (\ref{J}) can be truncated to a finite sum of terms and we evaluate the accuracy of the expression noting that the accuracy increases with the number of terms. However, it is noticed that when examining for 
different network parameters ($N$, $\rho$, $\theta$, $T_f$), the accuracy for expanding 1 term is 92.7$\%$, 2 terms is 99$\%$ and 99.9$\%$ for 3 terms only. Henceforth, in our analysis, 3 terms will be enough. Thus (\ref{Rayleigh}) provides a closed-form approximation for the EC in Rayleigh block fading when $M=2$. Moreover, the expectation in (\ref{EC}) is proved to be convex in $\epsilon$ \cite{paper5}. Thus, it has a unique minimizer $\epsilon^*$, which is consequently the EC maximizer. We define the optimum value of error probability $\epsilon^*$ which maximizes the EC in Rayleigh fading channels as
\begin{align}\label{e*}
\begin{split} 
\epsilon^*(\rho_i,c,d)= \arg\min_{0 \leq \epsilon \leq 1} \
 \epsilon+(1-\epsilon) \ \mathcal{J}.
\end{split}
\end{align}
To obtain the maximum effective capacity $EC_{max}$, we simply insert $\epsilon^*$ into (\ref{Rayleigh}).

\section{Performance analysis} \label{multinode}
We elaborate the effect of interference by plotting the per-node EC obtained from Lemma \ref{L1} for 1, 5 and 10 nodes in Fig. \ref{interference effect}. The network parameters are set as $T_f=1000, \rho=2$, and $\theta=0.01$ and the channel is assumed to be Rayleigh. It is obvious that the per-node EC decreases when increasing $N$ as more interference is added. Notice that the EC curves are concave in $\epsilon$ and hence, have a unique maximizer which is obtained from (\ref{e*}) and depicted in the figure. Another observation worth mentioning is that the optimum probability of error $\epsilon^*$ which maximizes the EC becomes higher when increasing the number of nodes. Notice that in Fig. \ref{interference effect}, we assume $M=2$ in (\ref{J}) which renders an accurate approximation to (\ref{EC}).

Given that all nodes transmit at the same time slot, the controller attempts decoding the transmitted symbols arriving from all of them. When the controller decodes one node's data, the other streams appear as interference to it \cite{paper3}. For this model, imagine that a node needs to raise its EC in order to meet its QoS constraint. We study the interference alleviation scenarios for one node at a certain time slot, while other nodes also keep transmitting at the same time.
\vspace{-2mm}
\subsection{Power control}\label{power_control}
The method of power control depends on increasing the SNR of node $n$ to allow it recover from the interference effect. Let $\rho_c$ be the new SNR of node $n$, while the other nodes still transmit with SNR equal to $\rho$. Then, we equate the SINR equation in (\ref{eq2.2}) to the case where no collision occurs ($N=1$) to obtain
\vspace{-5mm}
\begin{equation}\label{rhoc2}
\begin{split}
\rho_c&=\rho \ (1+\rho (N-1)).
\end{split}
\end{equation}

When a certain node transmits with SNR of $\rho_c$, its EC is the same as in the case when transmitting with SNR equals to $\rho$ while other nodes are silent. The method of power control is simple; however, it causes extra interference into other nodes due to the power increase of the recovering node. From (\ref{rhoc2}), we define the SINR of other nodes colliding in the same network (nodes in set $s$) after the compensation of one node as 
\begin{align}\label{eq44}
\rho_s&=\frac{\rho}{1+\rho_c+\rho (N-2)}=\frac{\rho}{1+\rho \ (\rho+1) (N-1)} 
\end{align}
\begin{figure}[!t] 
	\centering
	\includegraphics[width=0.9\columnwidth]{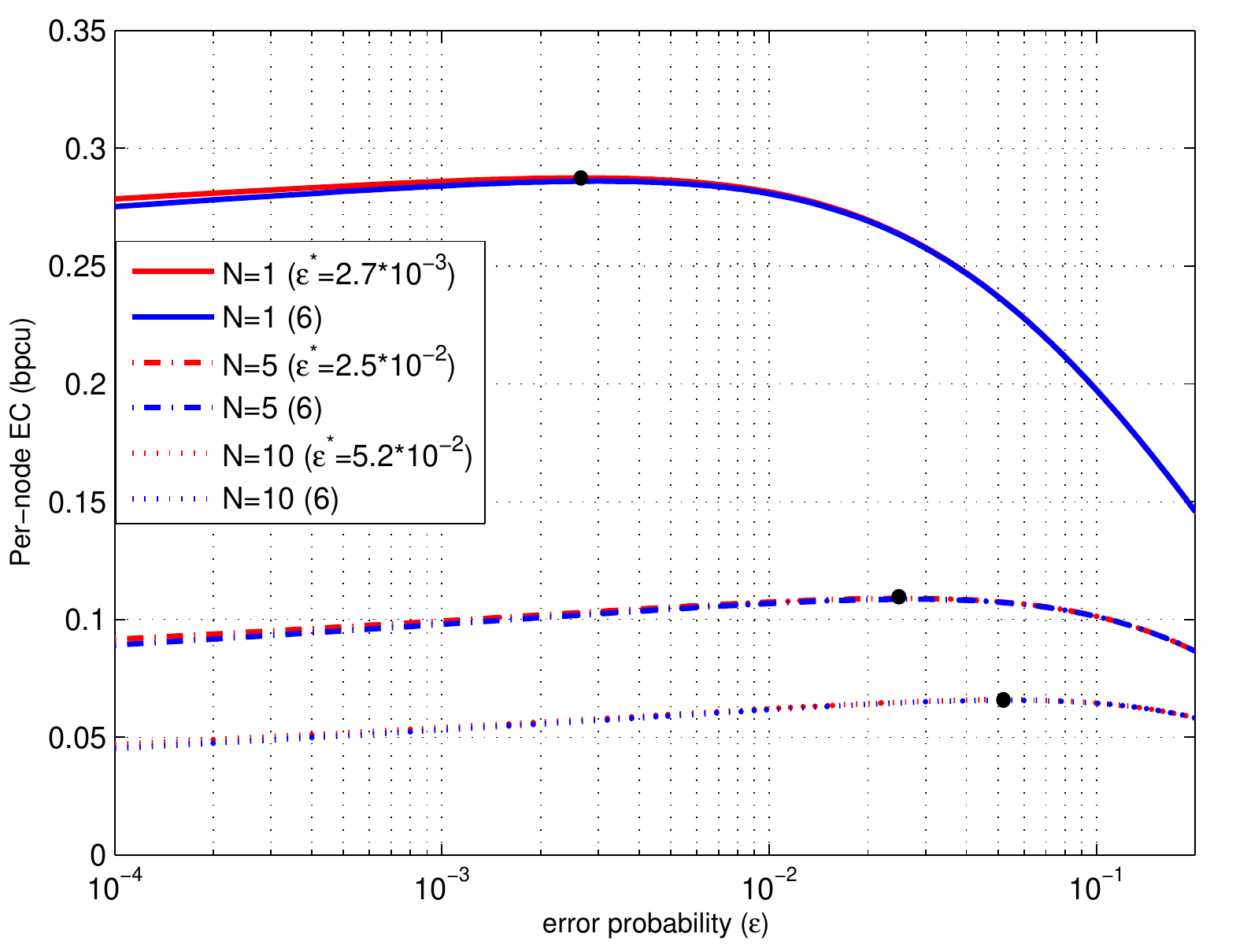}
	\centering
	\vspace{-4mm}
	\caption{EC as a function of error outage probability $\epsilon$ for different number of nodes, with $T_f=1000, \rho=2$, and $\theta=0.01$. }
	\label{interference effect}
	\vspace{-5mm}
\end{figure}
To determine the effect of compensation of one node on the other nodes, we define the compensation loss factor $\alpha_{c}$ as the ratio between maximum EC of other nodes (set $s$) in case of one node compensation and in case of no compensation. That is
\begin{align}\label{alphac}
\alpha_{c}=\frac{EC(\rho_s,\theta,\epsilon_s^*)}{EC(\rho_i,\theta,\epsilon_i^*)} 
\end{align}
where $\epsilon_s^*$ is the optimum error probability obtained from (\ref{e*}) when the SINR is set to $\rho_s$. To understand the effect of increased interference on the network performance, we study the effect of SINR variations on EC for different delay constraints.
\begin{proposition} \label{p1}
	SINR variations have comparably limited effect on EC when the delay constraint becomes more strict and vice versa.
\end{proposition}
\begin{proof}
	Differentiating (\ref{EC}) with respect to $\rho_i$
	\begin{align*}\label{}
	\begin{split} 
	\frac {\partial EC}{\partial \rho_i}=\frac {\partial EC}{\partial r}\frac {\partial r}{\partial \rho_i}=\frac{e^{-T_f \theta r}}{\epsilon+(1-\epsilon)e^{-T_f\theta r}} \mathcal{K},
	\end{split}
	\end{align*}
	where $\mathcal{K}= \frac {\partial r}{\partial \rho_i} (1-\epsilon)$ is strictly positive since the achievable rate $r$ is an increasing function of the SINR $\rho_i$. Differentiating once more with respect to $\theta$
	\begin{align}\label{}
	\begin{split} 
	\frac{\partial}{\partial \theta}\left( \frac {\partial EC}{\partial \rho_i}\right) = -\frac{\mathcal{K}T_f r e^{-T_f \theta r}}{\left(\epsilon+(1-\epsilon)e^{-T_f\theta r} \right)^2 },
	\end{split}
	\end{align}
	which is strictly negative and thus, validating our proposition.  
\end{proof}

Consider $T_f=1000$ and $\rho=1$, then Fig. \ref{alpha_c} depicts the compensation loss factor $\alpha_c$ for different number of nodes $N$ with $\theta=0.1$ and 0.001.  The figure shows that $\alpha_c$ is lower for smaller values of $\theta$. Hence, the effect of compensation appears to be more severe for less stringent delay constraints. This follows from Proposition \ref{p1} where SINR variations have less of an effect on delay strict networks and vice versa. Finally, we notice that the compensation loss factor decreases rapidly for a less dense network.
\begin{figure}[!t] 
	\centering
	\includegraphics[width=0.9\columnwidth]{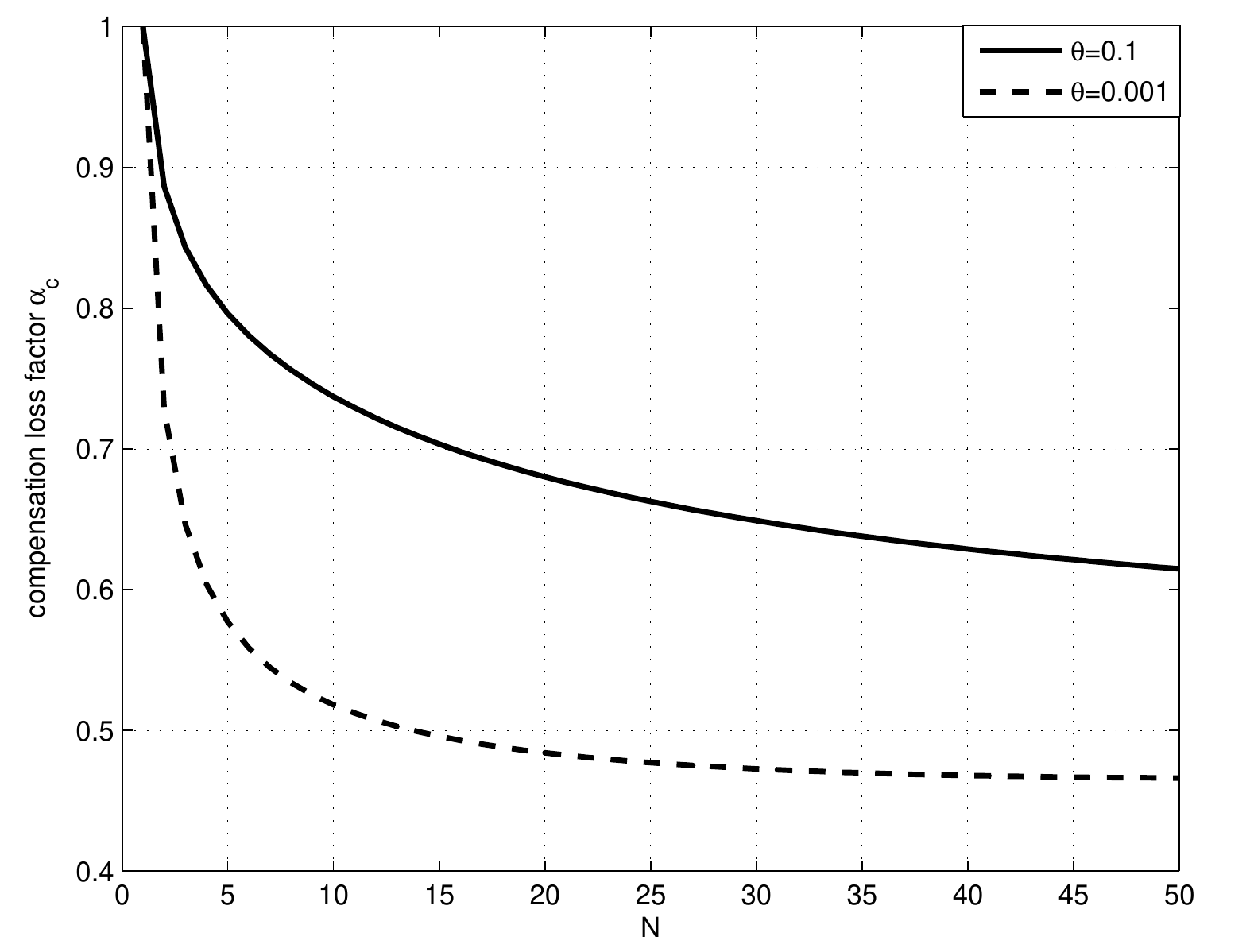}
	\centering
	\vspace{-3mm}
	\caption{Compensation loss factor $\alpha_c$ vs number of nodes $N$}
	\label{alpha_c}
	\vspace{-5mm}
\end{figure}
\subsection{Graceful degradation of the delay constraint} \label{theta comp}
Here we determine how the delay exponent $\theta$ should be gracefully degraded to obtain the same $EC_{max}$ as if the target node was transmitting without collision. This represents the cases where a node has flexible QoS constraint delay wise, so that the EC could be attained given a slight variation on the overall delay as envisioned in \cite{paper1}. Let $\theta$ be the original delay exponent and $\theta_i$ represent the new gracefully degraded one; $\theta_i$ is obtained by solving
\begin{align}\label{RE}
EC(\rho,\theta,\epsilon^*)=EC(\rho_i,\theta_i,\epsilon_i^*)
\end{align} 
where $\epsilon_i^{*}$ is the maximizer of EC for the parameters $(\rho_i,\theta_i)$ and $\epsilon_i^*$ is the optimum error probability for $(\rho_i,\theta_i)$. The solution of (\ref{RE}) renders the necessary value of $\theta_i$ to compensate for the EC decrease due to collision in this case. In Fig. \ref{theta compensation}, we illustrate the graceful degradation of the delay constraint as a function of the target error outage probability. Consider (\ref{RE}) with $N=5$, $\theta_1=0.05$, $\rho=1$ and $T_f=1000$, we get $\theta_i=0.023$. Thus, by gracefully degrading the delay constraint from $0.05$ to $0.023$, we attain the same value for the maximum effective capacity $EC_{max}=0.066$. For a delay outage probability of $10^{-3}$, this corresponds to extending the allowable delay $D_{max}$ from 3600 to 4600 symbol periods. We perform a limited delay extension ($\approx 25 \%$) because the rise in EC partially compensates the graceful degradation of $\theta$ in (\ref{delay}). Note that the optimum error probabilities have different values in each case due to the change in SINR in (\ref{J}).   

\begin{figure}[!t] 
	\centering
	\includegraphics[width=0.9\columnwidth]{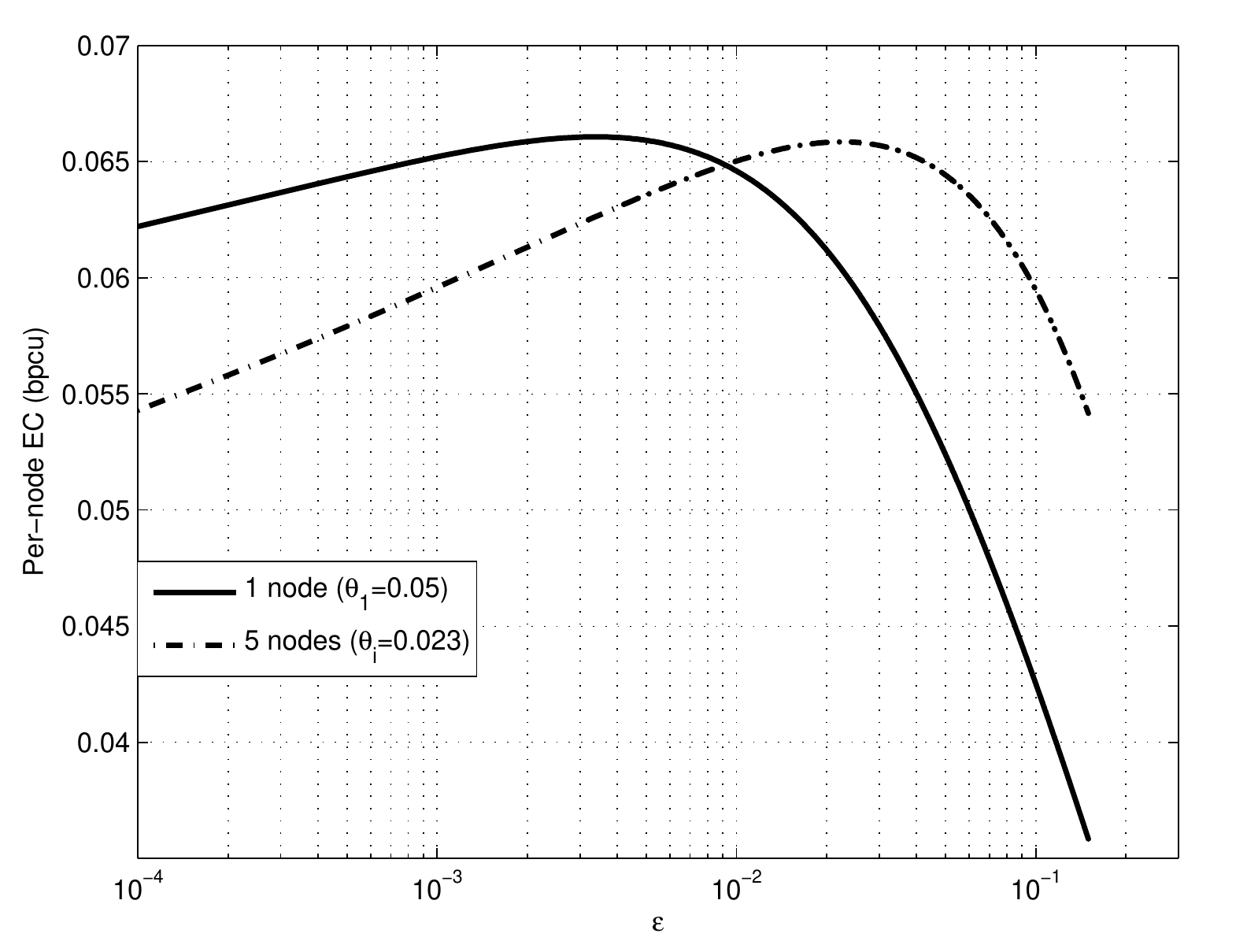}
	\centering
	\vspace{-4mm}
	\caption{Graceful degradation delay constraint $\theta$ in case of 5 nodes colliding where $T_f=1000$ and $\rho=1$.}
	\label{theta compensation}
	\vspace{-4mm}
\end{figure}

\subsection{Joint compensation model} \label{joint compensation} 
To mitigate the side effects of power control and graceful delay constraint degradation, we apply a joint model in which both methods are partially employed. Define the operational SINR in power controlled compensation for nodes in set $s$ as $\rho_{s_o}$, where $\rho_{s_o}$ lies on the interval [$\rho_s$ $\rho_i$]. Using (\ref{eq44}), the operational SNR for the recovering node can be written as
\begin{equation}\label{roco}
\rho_{c_o}=\frac{\rho}{\rho_{s_o}} -1-\rho(N-2),
\end{equation} 
and the operational point of the compensation loss factor $\alpha_{c_o}$ is
\begin{equation}\label{alphaco}
\alpha_{c_o}=\frac{EC(\rho_{s_o},\theta,\epsilon_{s_o}^*)}{EC(\rho_i,\theta,\epsilon_i^*)} 
\end{equation}
where $\epsilon_{s_o}^*$ is the optimum error probability  obtained from (\ref{e*}) for the parameters ($\rho_{s_o},\theta_1$). $\alpha_{c_o}$ is considered to be the loss factor caused by the part of compensation performed via power control. 

Next, we perform the rest of compensation via graceful degradation of $\theta$ as in Section \ref{theta comp}. To obtain $\theta_2$, we solve
\begin{align}\label{theta joint}
&EC(\rho,\theta,\epsilon^*)=EC(\frac{\rho_{c_o}}{1+\rho (N-1)},\theta_2,\epsilon_2^*) 
\end{align}    
From (\ref{theta joint}), we compute the necessary value of $\theta_2$ to continue the compensation process via graceful degradation of the delay constraint.
\begin{figure}[!t] 
	\centering
	\includegraphics[width=0.9\columnwidth]{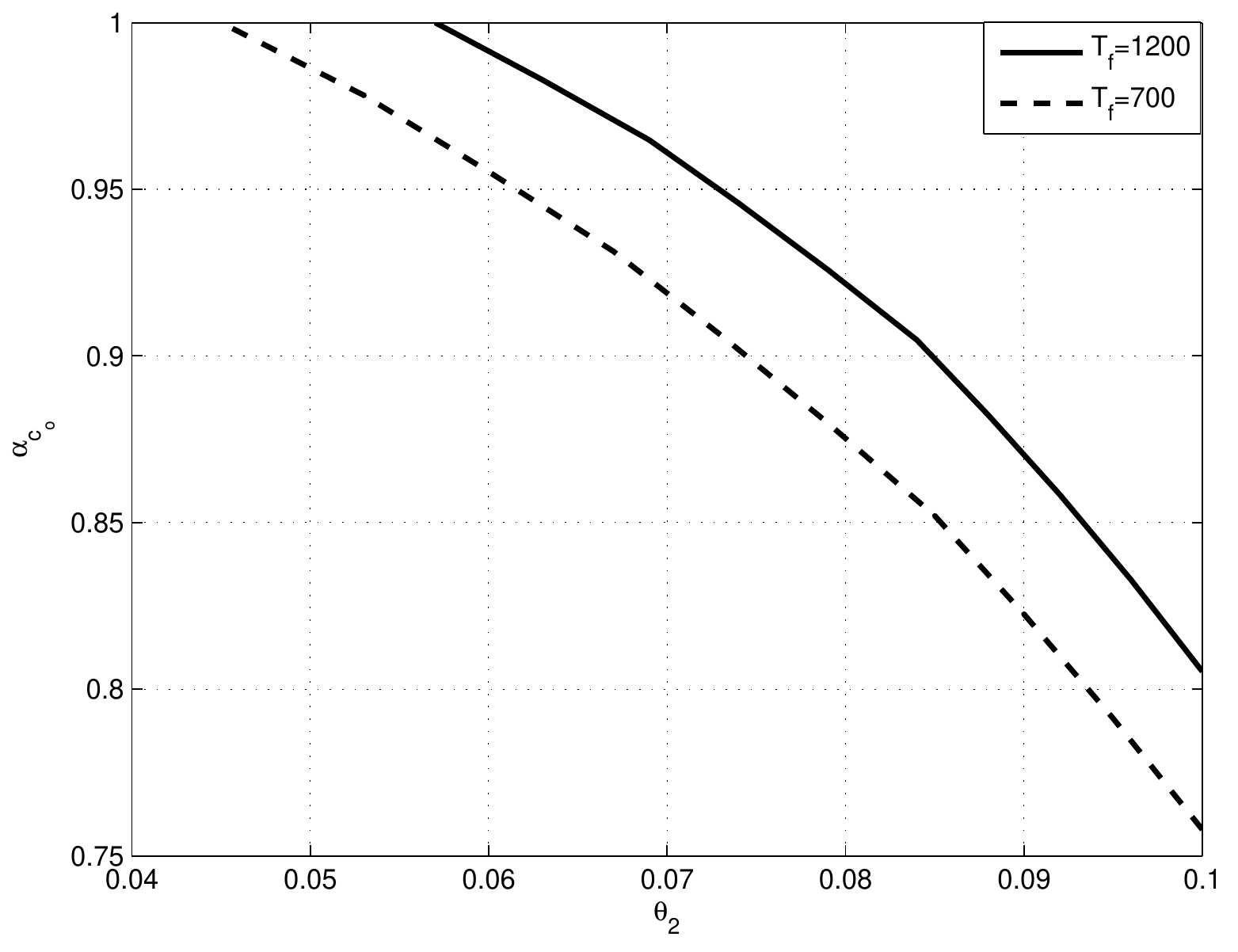}
	\centering
	\vspace{-4mm}
	\caption{Trade off between compensation loss factor via power control $\alpha_{c_o}$ and graceful degradation of delay constraint $\theta_{2}$ for different operational points}
	\label{trade off}
	\vspace{-5mm}
\end{figure}

Fig. \ref{trade off} illustrates different operational points for the joint model for different blocklength $T_f$ where $N=5$, $\rho=1$, $P_{out\_ delay}=10^{-3}$ and $\theta_1=0.1$. For example, when $T_f=700$, we select the operational point $\alpha_{c_o}=0.9, \theta_2=0.075$. This implies that a part of compensation will be performed via power control, which leads to $10\%$ loss in EC of other nodes (set $s$). Then, the rest of compensation will be performed by gracefully degrading its $\theta$ from 0.1 to 0.075. The maximum delay of the recovering node remains 2500 symbol periods before and after recovery as restoring the EC compensates for the decrease in $\theta$ in (\ref{delay}). The figure also shows that for smaller packet sizes, the amount of losses due to compensation are higher. 

Now, we propose an objective function leveraging the network performance for the joint model. First, we define the priority factor $\eta_{\alpha}$ as a measure of the risk of decrease in EC of nodes in set $s$ when the compensating node boosts its transmission power. In other words, the higher the value of $\eta_{\alpha}$, the more important it is not to allow much degradation of EC of nodes in set $s$ and hence, we try not to compensate via power control and shift compensation towards $\theta$ graceful degradation. On the other hand, we define the priority factor $\eta_\theta$ as a measure of strictness of the delay constraint (i.e., the higher the value of $\eta_\theta$, the more strict it is not to degrade delay constraint). Thus, we can formalize our objective function as the summation
\begin{equation}\label{eta}
\eta=\eta_\alpha \alpha_{c_o}+\eta_\theta \theta_2
\end{equation} 
where ($\alpha_{c_o},\theta_2$) is the operational point. Now, we choose this operational point to satisfy 
\begin{equation}\label{op}
\begin{split}
\eta_{max}=&\max_{\theta_2 \geq 0} \ \eta_\alpha \alpha_{c_o}+\eta_\theta \theta_2 \\ 
s.t \ \ & \rho_s \leq \rho_{s_o} \leq \rho_i \\
\end{split}
\end{equation} 
where the solution to this problem gives the optimum operational point which can be found from (\ref{roco}), (\ref{alphaco}) and (\ref{theta joint}).

For an MTC network with 15 devices where $T_f=1000, \rho=2, \theta_1=0.1,\eta_\alpha=1$ and $\eta_\theta=4$, the optimum value of $\rho_{s_o}$ will be 0.057. This value corresponds to the operational point $\alpha_{c_o}=0.9397$ and $\theta_2=0.053$. The SNR of the recovering node becomes $\rho_{c_o}=8.08$. Thus, to maximize the network throughput according to the given priority factors, the compensating node boosts its SNR from 2 to 8.08 and gracefully degrades its delay exponent from 0.1 to 0.053. This results in only 6 \% loss in EC of other nodes as depicted in Fig. \ref{another}. 
\begin{figure}[!t] 
	\centering
	\includegraphics[width=0.9\columnwidth]{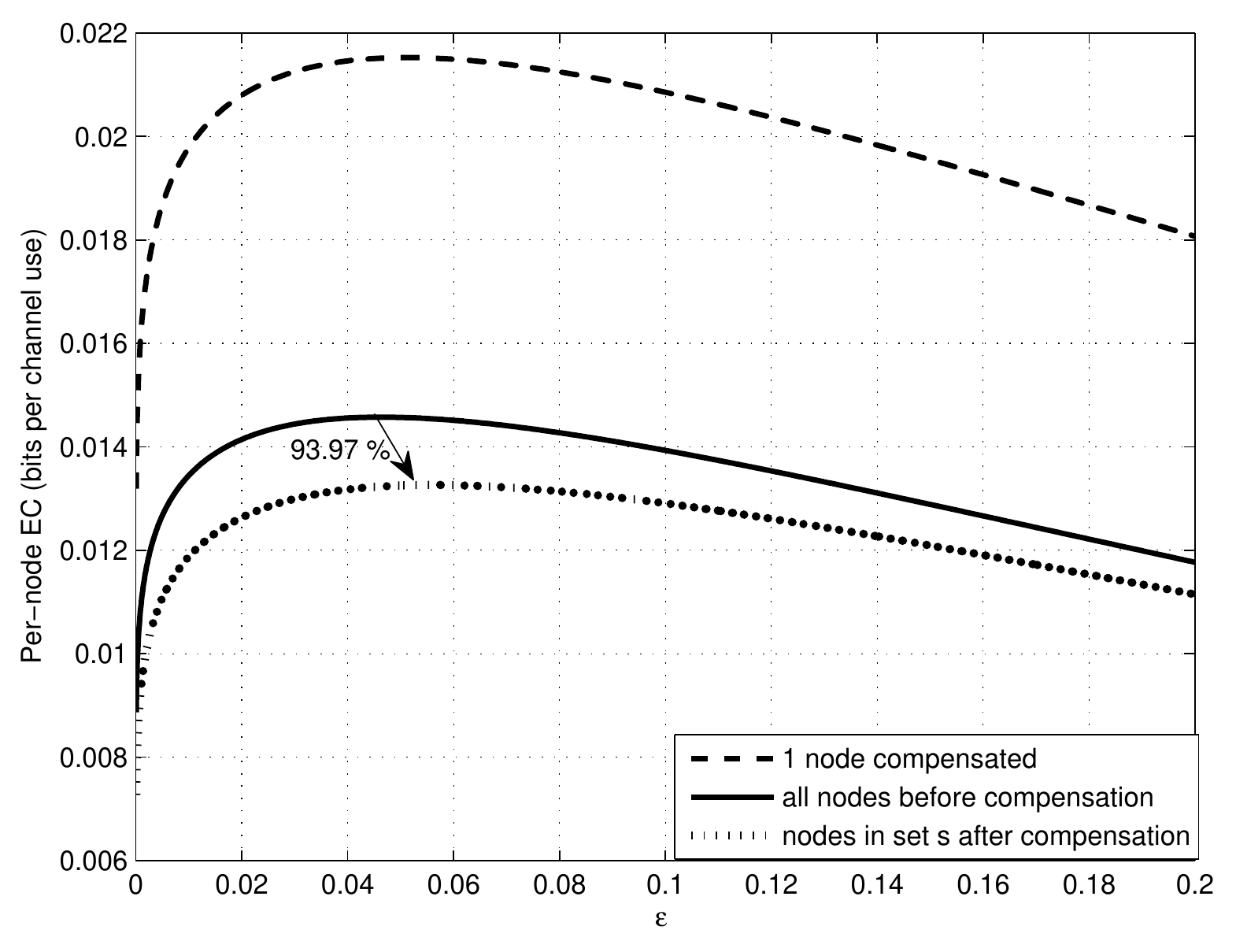}
	\centering
	\vspace{-5mm}
	\caption{Per-node EC as a function of error probability $\epsilon$ before and after joint compensation for $T_f=1000, \theta_1=0.1, \rho=2$, and $N=15$}
	\label{another}
	\vspace{-5mm}
\end{figure}
\section{Conclusion} \label{conclusion}
\vspace{-1mm}
In this work, we presented a detailed analysis of the EC for delay constrained MTC networks in the finite blocklength regime. For Rayleigh block fading channels, we proposed an approximation for the EC and characterized the optimum error probability. Our analysis indicated that SINR variations have minimum effect on EC under strict delay constraints. In a dense MTC network scenario, we illustrated the effect of interference on EC. We proposed power control as an adequate method to restore the EC in networks with less stringent delay constraints. Another method is graceful degradation of delay constraint, where we showed that a very limited extension in delay limit could successfully recover the EC. Joint compensation emerges as a combination between these two methods, where an operational point is selected to maximize an objective function according to the networks design aspects. As future work, we aim to analyze the impact of imperfect CSI on the EC and coordination algorithms that maximize EC with fairness constraints.
\vspace{-1.5mm}
\section*{Acknowledgments}
\vspace{-1mm}
This work has been partially supported by Finnish Funding Agency for Technology and Innovation (Tekes), Huawei Technologies, Nokia and Anite Telecoms, and Academy of Finland (under Grant no. 307492).
\vspace{-2mm}
\bibliographystyle{IEEEtran}
\bibliography{di}
\end{document}